\documentclass[submission,copyright,noderivs]{eptcs}
\providecommand{\event}{NCMA 2026} 

\usepackage{mathrsfs}
\usepackage{amssymb}
\usepackage{amsmath}
\usepackage{amsthm}
\usepackage{textcomp}
\usepackage{pifont}
\usepackage[usenames,dvipsnames]{xcolor}
\usepackage{colortbl}
\usepackage{transparent}

\theoremstyle{plain}
\newtheorem{theorem}{Theorem}
\newtheorem{proposition}[theorem]{Proposition}
\newtheorem{lemma}[theorem]{Lemma}
\newtheorem{corollary}[theorem]{Corollary}

\theoremstyle{definition}

\newcommand{\lfam}{\mathscr{L}}

\newcommand{\border}{\texttt{\#}}

\newcommand{\rightend}{\mathord{\vartriangleleft}}
\newcommand{\blank}{\raisebox{\depth}{\texttt{\char 32}}}

\newcommand{\accept}{\texttt{accept}}

\newcommand{\dfa}{\textsf{DFA}}
\newcommand{\dfawtl}{\textsf{DFAwtl}}
\newcommand{\nrdfawtl}{\textsf{nrDFAwtl}}
\newcommand{\nfa}{\textsf{NFA}}
\newcommand{\nfawtl}{\textsf{NFAwtl}}
\newcommand{\nrnfawtl}{\textsf{nrNFAwtl}}

\newcommand{\ntrans}{\textsf{NIUFST}}

\newcommand{\subtext}[1]{\textnormal{\scriptsize #1}}
\newcommand{\tcyes}{\textcolor{PineGreen}{\transparent{0.4}\ding{51}}}%
\newcommand{\tcno}{\textcolor{red}{\transparent{0.4}\ding{55}}}%
\newcommand{\cno}{\textcolor{red}{\ding{55}}}%
\newcommand{\copen}{\textcolor{blue}{\textbf{?}}}
\title{On some Open Problems for Finite Automata\\ with Translucent Input Letters}

\author{Martin Kutrib \and Andreas Malcher \and Matthias Wendlandt
\institute{%
  Institut f{\"u}r Informatik, Universit{\"a}t Giessen\\
  Arndtstr.~2, 35392 Giessen, Germany}
\email{$\{$kutrib, andreas.malcher, matthias.wendlandt$\}$@uni-giessen.de}
}
 
\def\titlerunning{On some Open Problems for Finite Automata with Translucent Letters}
\def\authorrunning{M.~Kutrib, A.~Malcher, M.~Wendlandt}

\begin{document}

\maketitle

\begin{abstract}
Finite automata with translucent input letters are a recent model of discontinuous input processing. Basically,
classical finite automata are equipped with a translucency function that defines, depending on the state, the
set of translucent input symbols. While processing the input, translucent symbols are skipped and only visible symbols
are read and processed. It is distinguished between deterministic and nondeterministic models and, in addition,
between returning and non-returning models. In the former case,
the automaton restarts from the left end of the input after having consumed some visible symbol, whereas in the
latter case the automaton restarts from the left end of the input when the right endmarker symbol is seen.

Returning finite automata with translucent letters have been introduced by Nagy and Otto and its non-returning
variant has been introduced by Mr{\'{a}}z and Otto. Many results concerning the computational capacity, relations
between deterministic and nondeterministic models, and relations between returning and non-returning models are known.
Moreover, some results on closure properties and decidability questions have been obtained as well. However, some
questions have still been open since many years. In this paper, we will give answers to some of these open questions.

In particular, we show the non-closure under concatenation, Kleene star, reversal, and inverse homomorphism for the 
non-returning deterministic as well as nondeterministic model. We also obtain non-closure under inverse homomorphism
for the returning deterministic and nondeterministic model. Finally, we investigate the emptiness problem for non-returning finite
automata with translucent input letters and show the decidability of the problem in case of deterministic as well as
nondeterministic automata.
\end{abstract}

\section{Introduction}

An automaton model typically processes its input in a continuous way, namely from left to right in models with
one-way motion of the input head.
In addition, one symbol is read at one time step and the input is accepted or rejected when the end of the input is reached.
This ``standard'' input mode has been extended in the literature in several directions. For example, 
two-way motion of the input head, stationary moves of the input head, rotating the input head, or restarting modes 
have been considered for a wide variety of machines (the reader is referred to, e.g., 
\cite{bensch:2009:irdnfa,Hopcroft:1979:itatlc:book,Jancar:1995:ra:proc,otto:2025:book} for an overview of these and other models).
In all these extensions studied, consequences on the computational and descriptional power of the machines could be observed.
Thus, the way of processing the input can be considered as a computational resource that then can be used to tune 
computational and descriptional power.

In recent literature several variants of \emph{discontinuous} input processing have been introduced and investigated.
One of these variants is the ``jumping'' paradigm which means that jumping to any position inside the input string is allowed at any move. 
This paradigm has been introduced for finite automata in~\cite{meduna:2012:jfa}, where it is shown that this discontinuous way of
input processing may increase the computational power since it allows to accept some non-context-free languages,
but also may decrease the computational power since it is known that some regular languages are not accepted by jumping finite automata~\cite{meduna:2012:jfa}.
``Right one-way jumping'' automata which are a restricted variant of jumping finite automata are investigated in depth in~\cite{beier:2022:nrowjfa,chigahara:2016:owjfa}.
Another discontinuous way of input processing is to make some parts of the input, depending on the state, \emph{invisible}. 
This concept of \emph{translucent letters} has been introduced by Nagy and Otto in~\cite{nagy:2011:fsawtl} for deterministic and nondeterministic finite automata
and the basic idea for such devices is to provide a translucency function that defines in which states which letters of the input are translucent. 
At each move, the device skips (by looking through) the translucent portion of the input, from the
current input head position up to the first non-translucent letter (thus realizing a jump).
After processing the non-translucent symbol, in the \emph{returning mode} the
input head returns to the left end of the input while, in the \emph{non-returning mode}, the input head continues to process the input 
according to its updated current state and the corresponding translucent symbols. In both modes, the input head returns to the left end when the right
end of the input is reached.  

Deterministic and nondeterministic finite automata with translucent input letters have deeply been investigated in the literature 
(see, e.g.,~\cite{mraz:2023:nrfawtl,nagy:2012:oncdsosdrawwso,otto:2023:asoawtl:proc}). However, many questions are still open,
but some of them will be answered in the present paper. Automata with translucent letters are an active area of research
and many extensions and variants of the original model are currently
investigated. Recently studied extensions are finite
automata with translucent words~\cite{Nagy:2024:fawsotw} and finite automata with translucent letters and two-way motion of
the input head~\cite{kutrib:2025:twfawtil:proc}. Pushdown automata with translucent input letters in the returning mode are originally studied
in~\cite{nagy:2011:cdsosdragbaeps,Nagy:2013:dpcdsosdra} in terms of certain cooperating distributed systems of restarting automata with 
additional pushdown store. However, in their model $\lambda$-transitions are not allowed and acceptance is defined by empty pushdown.
General pushdown automata with translucent input letters working in the returning mode as well as in the non-returning mode
are introduced and studied in~\cite{kutrib:2026:dpawtil,kutrib:2026:opolabdpawtil}.
In~\cite{Otto:2015:ovptl} the paradigm of \emph{input-driven} languages was imposed to certain cooperating distributed systems of restarting automata 
with additional pushdown store with regard to characterizing certain trace languages. A general definition of 
input-driven pushdown automata with translucent input letters and many results concerning their computational capacity, closure
properties, decidability questions has been given in~\cite{kutrib:2025:idpawtil}. The \emph{jump complexity} of finite automata
with translucent input letters is studied in~\cite{jcodfawtl:2024:proc,mitrana:2024:jcofawtl}, 
whereas subsequence matching problems of such automata are investigated in~\cite{smaapfawtl:2025:proc}.

In this paper, we are going to investigate some problems on returning and non-returning finite automata with
translucent input letters that have been left open in~\cite{mraz:2023:nrfawtl,nagy:2012:oncdsosdrawwso}. These open problems
concern closure properties as well as decidability questions and will be answered in the present paper
which is organized as follows. After giving the necessary definitions in the next section, 
we study in Section~\ref{sec:open-closure} several open problems on closure properties. 
For non-returning deterministic and nondeterministic finite automata with translucent input letters we can prove
that the corresponding language families are not closed under concatenation, Kleene star, inverse homomorphism,
and reversal. These questions have been left open in~\cite{mraz:2023:nrfawtl}.
We further obtain the non-closure under complementation for non-returning nondeterministic finite automata with translucent input letters
and the non-closure under inverse homomorphism for returning deterministic and nondeterministic finite automata with translucent input letters.
The latter question has been left open in~\cite{nagy:2012:oncdsosdrawwso}. In Section~\ref{sec:open-decide},
we look at the emptiness problem for deterministic and nondeterministic finite automata with translucent input letters.
It is known that this problem is decidable for the returning variants, whereas its decidability status for the non-returning case
is left open in~\cite{mraz:2023:nrfawtl}. Here, we can show that the emptiness problem is also decidable in the non-returning case
for deterministic as well as nondeterministic machines.

\section{Preliminaries}\label{sec:prelim}

We denote by $\Sigma^*$ the set of all words on the finite alphabet $\Sigma$, including the 
empty word $\lambda$, and let
$\Sigma^+ = \Sigma^* \setminus \{\lambda\}$. For any word $w\in\Sigma^*$, we let 
$|w|$ denote its length, $w^R$ its reversal, and $|w|_a$ the number of occurrences of the 
symbol $a\in\Sigma$ in $w$. 
We use $\subseteq$ for {inclusions}, and $\subset$ for proper inclusion. 
Given a set~$S$, we denote by $2^S$ its power set, and by $|S|$ its cardinality. 
A language on  $\Sigma$ is any subset $L \subseteq \Sigma^*$. The complement of $L$
is the language $\overline{L}=\Sigma^* \setminus L$, its reversal is $L^R=\{\,w^R \mid w\in L\,\}$.
Two language families~$\mathscr{L}_1$ and~$\mathscr{L}_2$
are said to be {incomparable} whenever~$\mathscr{L}_1$ is not a subset of~$\mathscr{L}_2$
and vice versa.

Finite automata with translucent letters are extensions of classical
finite automata that do not have to read their inputs from left to right.
Instead, depending on the current state of such devices, some of the input
letters may be translucent (invisible). Accordingly, a finite automaton
with translucent letters reads and processes the first visible input letter
from left. The following non-returning finite automata with translucent
letters have been introduced in~\cite{mraz:2023:nrfawtl}.

A \emph{non-returning nondeterministic finite automaton with translucent letters} (\nrnfawtl) is defined as a system
$M=\langle Q,\Sigma, q_0, \rightend, \tau,\delta\rangle$, 
where
$Q$ is the finite set of {states}, $\Sigma$ is the finite set of {input symbols}, with
$\Sigma \cap Q=\emptyset$,
$q_0 \in Q$ is the {initial state},
$\rightend\notin \Sigma$ is the {endmarker},
$\tau\colon Q\to 2^\Sigma$ is the {translucency mapping}, and
$$
\delta\colon Q\times (\Sigma\cup\{\rightend\}) \to 
2^Q \cup\{\accept\}
$$ 
is the {partial transition function}.
The translucency mapping $\tau$ bear the following meaning:
for any state $q\in Q$, the
letters from the set $\tau(q)$ are \emph{translucent (invisible) for $q$}, that is, whenever in $q$, the automaton~$M$ does
not see these letters (or equivalently, $M$ can see through such letters).

A {configuration} of $M$ is represented by $(uqv\rightend)$ or $\accept$, 
where $q \in Q$ is {the current state}, $uv\in \Sigma^*$ is the remaining tape inscription
with $u$ being to the left and $v$ to the right of the input head.
The {initial configuration} of $M$ on input $w\in\Sigma^*$ is $(q_0w\rightend)$.
Along its computation, $M$ runs through a sequence of configurations. Being in
some configuration $(uqv\rightend)$, a step of $M$ takes place
as follows. First, $M$ determines the next input symbol to be processed
by taking the first letter to the right of the input head that is visible in
state $q$. Precisely, if~$v=xay$ with $x\in \tau(q)^*$ and $a\notin \tau(q)$, then $M$ takes~$a$. Now, 
$M$ halts and rejects whenever $\delta(q,a)$ is undefined, otherwise it chooses
some state $p\in \delta(q,a)$ and the successor configuration is
$(uxpy\rightend)$.
In the case $v\in \tau(q)^*$, $M$ sees the endmarker~$\rightend$,
and it halts and accepts if and only if $\delta(q,\rightend)=\accept$. 
One step from a configuration to its {successor configuration} is denoted
by~$\vdash$, which is specified as follows.
Let
$p,q\in Q$, $a \in \Sigma$, and $x,y,u,v \in \Sigma^*$. Then:
\begin{enumerate}
\itemsep 1mm
\item 
  $(uqxay\rightend) \vdash (uxpy\rightend)$,
  if $x\in \tau(q)^*$, $a\notin \tau(q)$, and $p\in \delta(q,a)$,
\item 
 $(uqv\rightend) \vdash (puv\rightend)$,
  if $v\in \tau(q)^*$ and $p\in \delta(q,\rightend)$,
\item 
 $(uqv\rightend) \vdash \accept$,
  if $v\in \tau(q)^*$ and $\accept = \delta(q,\rightend)$.
\end{enumerate}

We let $\vdash^*$ (resp.,~$\vdash^+$) denote the reflexive and transitive
(resp., transitive) closure of~$\vdash$.
Sometimes, we will be saying that an $\nrnfawtl$ performs \emph{sweeps}, where a sweep 
is a sequence of transitions that starts with the input head at the left end of the
(remaining) input and ends after the next (if any) return move on the
endmarker (move of type (2) above).

The language accepted by the $\nrnfawtl$ $M$ is the set $L(M)$ of those words in~$\Sigma^*$ for which
the computation, beginning in the initial
configuration, eventually halts accepting, namely:
$$
L(M) = \{\,w \in \Sigma^* \mid (q_0w\rightend)\vdash^+ \accept \,\}.
$$
In general, we denote by $\lfam({\sf X})$ the family of all languages accepted by some 
device $\sf{X}$.

The notion of translucent input letters in the realm of finite automata has
first been introduced in~\cite{nagy:2011:fsawtl}, where the finite automata
work in the so-called returning mode.
According to this paradigm,
after processing a visible letter, the head of the automaton always returns to
the beginning of the input (abbreviated as $\nfawtl$). Basically, definitions and notations
for $\nrnfawtl$ apply to $\nfawtl$ as well. The successor configurations are
defined correspondingly to implement the mode.

Deterministic (non-returning) finite automata with translucent letters
($\nrdfawtl$, $\dfawtl$) are special cases of $\nrnfawtl$ and $\nfawtl$, where
$|\delta(q,a)|\leq 1$ for all $q\in Q$, $a\in\Sigma\cup\{\rightend\}$. In
particular, this definition implies that $\dfawtl$ may continue a
computation by another sweep when the endmarker is seen.
Originally, $\dfawtl$ have been introduced~\cite{nagy:2011:fsawtl} and 
studied~\cite{nagy:2013:gdcdssdr} as possible
interpretation of cooperating distributed systems of certain restarting automata.
Because of this interpretation, $\dfawtl$ and $\nfawtl$ have to halt when the
endmarker is seen. While this condition does not affect the computational
capacity of $\nfawtl$, it critically limits the computational power 
of $\dfawtl$~\cite{mraz:2024:rfawtl:proc}. To move beyond the limitations of
that particular interpretation, several models in the literature are 
studied with translucent input symbols as natural extensions of the classical
models, that may continue the computation when the endmarker is seen.
Note that such automata are sometimes called repetitive in the literature.  
However, here we will continue to use the notion $\dfawtl$ for the general
model without calling it repetitive.

\section{Open Closure Properties}\label{sec:open-closure}

In this section, we turn to solve some open closure properties of the language
families defined by returning and non-returning variants of deterministic and
nondeterministic finite automata with translucent input letters. We start with
an observation that gives the non-closure under complementation for non-returning
nondeterministic finite automata with translucent input letters.

The language $L_\subtext{eq} = \{\, w\in \{a,b\}^+ \mid |w|_a=|w|_b\,\}$
is accepted by some $\dfawtl$~\cite{Nagy:2013:dpcdsosdra}.
So, it is accepted by some $\nrdfawtl$ and some $\nrnfawtl$.
Now, we consider the regular language
$R=\{aa,bb\}^*$ and build the intersection \mbox{$L_\subtext{dab}=
  L_\subtext{eq}\cap R$.} 
The language $L_\subtext{dab}$ has been introduced
in~\cite{kutrib:2025:twfawtil:proc}. In~\cite{kutrib:2026:twfawtl} it is shown
$L_\subtext{dab}$ is not accepted by any $\nrnfawtl$. So, neither the family
$\lfam(\nrnfawtl)$ nor the family $\lfam(\nrdfawtl)$ is closed under
intersection with regular languages. Since both are containing strictly the
family of regular languages, it is a trivial implication that
then both families are not closed under intersection. Since
$\lfam(\nrnfawtl)$ is closed under union~\cite{mraz:2023:nrfawtl},
it cannot be closed under complementation by De Morgan's laws.

\begin{corollary}\label{cor:nrnfawtl-not-closed-compl}
The family $\lfam(\nrnfawtl)$ is not closed under complementation.
\end{corollary}

\subsection{Non-Returning Computations and Concatenation, Kleene Star and Reversal}\label{subsec:concat-star}

The question of whether the families $\lfam(\nrdfawtl)$ and
$\lfam(\nrnfawtl)$ are closed under concatenation and Kleene star are presented as
open problems in~\cite{mraz:2023:nrfawtl}. In addition, their closure under
reversal has been left open. To solve these problems, we consider the language
$L_c = L_{c,1} \cdot L_{c,2}$, where
$$
L_{c,1}=\{\, a^n\border^n a^n\mid n\geq 0\,\}
\text{ and } L_{c,2}=\{\, \border v\mid v \in \{a,b\}^+, |v|_a=|v|_b\,\}.
$$

\begin{lemma}\label{lem:lc-not-nrnfawtl}
The language $L_c$ is not accepted by any $\nrnfawtl$.
\end{lemma}

\begin{proof}
\begin{sloppypar}
Assume in contrast to the assertion that $L_c$ is accepted by some 
$\nrnfawtl$ $M=\langle Q,\Sigma, q_0, \rightend, \tau,\delta\rangle$.
We consider an accepting computation on input
$a^n\border^n a^n \border a^mb^m$, where $m,n$ are large enough.
\end{sloppypar}

Suppose that $M$ reads more than $|Q|$ many symbols $a$, say $n_1$, from the
prefix of the input without skipping any translucent letters. Then, during the first $n_1$ steps one state, say $p_1$,
appears at least twice:
$$
(q_0 a^n\border^n a^n \border a^mb^m \rightend)
\vdash^* 
(p_1 a^{n-i_1} \border^n a^n \border a^mb^m \rightend)
\vdash^+
(p_1 a^{n-i_1-i_2} \border^n a^n \border a^mb^m \rightend),
$$
where $i_1,i_2\leq |Q|$.
Then, the computation 
$$
(q_0 a^{n-i_2}\border^n a^n \border a^mb^m \rightend)
\vdash^* 
(p_1 a^{n-i_2-i_1} \border^n a^n \border a^mb^m \rightend)
$$
continues accepting as well, a contradiction. So, $M$ jumps after at most
$|Q|$ steps on the prefix. If it jumps to the endmarker, the effect is a state
change only and we continue the argumentation in the same way as jumping to
anywhere else.

\medskip
\noindent
\textbf{Case~1.} We assume that $M$ jumps where the symbol $\border$ is not translucent:
$$
(q_0 a^n\border^n a^n \border a^mb^m \rightend)
\vdash^* 
(q_1 a^{n-x_1} \border^n a^n \border a^mb^m \rightend)
\vdash
(a^{n-x_1} q_2 \border^{n-1} a^n \border a^mb^m \rightend).
$$
We can repeat the argument from above. If $M$ reads next more than $|Q|$
consecutive symbols $\border$ a state, say $p_2$,  appears twice:
\begin{align*}
(q_0 a^n\border^n a^n \border a^mb^m \rightend)
&\vdash^+
(a^{n-x_1} q_2 \border^{n-1} a^n \border a^mb^m \rightend)\\
&\vdash^*
(a^{n-x_1} p_2 \border^{n-j_1} a^n \border a^mb^m \rightend)\\
&\vdash^+
(a^{n-x_1} p_2 \border^{n-j_1-j_2} a^n \border a^mb^m \rightend),
\end{align*}
where $1\leq j_1,j_2\leq |Q|$.
Then, the computation 
$$
(q_0 a^n\border^{n-j_2} a^n \border a^mb^m \rightend)
\vdash^*
(a^{n-x_1} p_2 \border^{n-j_1-j_2} a^n \border a^mb^m \rightend)
$$
continues accepting as well, a contradiction. So, $M$ jumps after at most
$|Q|$ steps on the factor $\border^n$.

\medskip
\noindent
\textbf{Case~1.1.} We assume furthermore that $M$ now jumps where the 
symbol $a$ is not translucent:
$$
(q_0 a^n\border^n a^n \border a^mb^m \rightend)
\vdash^+ 
(a^{n-x_1} \border^{n-y_1} q_3 a^{n-1} \border a^mb^m \rightend).
$$
Once more, we repeat the argument from above. If $M$ reads more than $|Q|$
consecutive symbols $a$ a state, say $p_3$,  appears twice:
\begin{align*}
(q_0 a^n\border^n a^n \border a^mb^m \rightend)
&\vdash^+ 
(a^{n-x_1} \border^{n-y_1} q_3 a^{n-1} \border a^mb^m \rightend)\\
&\vdash^*
(a^{n-x_1} \border^{n-y_1} p_3 a^{n-k_1} \border a^mb^m \rightend)\\
&\vdash^+
(a^{n-x_1} \border^{n-y_1} p_3 a^{n-k_1-k_2} \border a^mb^m \rightend),
\end{align*}
where $1\leq k_1,k_2\leq |Q|$.
Then, the computation 
$$
(q_0 a^n\border^n a^{n-k_2} \border a^mb^m \rightend)
\vdash^+
(a^{n-x_1} \border^{n-y_1} p_3 a^{n-k_1-k_2} \border a^mb^m \rightend)
$$
continues accepting as well, a contradiction. So, $M$ jumps after at most
$|Q|$ steps on the second factor $a^n$.

\medskip
\noindent
\textbf{Case~1.1.1.} We assume furthermore that $M$ now jumps where the 
symbol $\border$ is not translucent:
$$
(q_0 a^n\border^n a^n \border a^mb^m \rightend)
\vdash^+ 
(a^{n-x_1} \border^{n-y_1} a^{n-z_1} q_4 a^mb^m \rightend).
$$
Once more, we repeat the argument from above and obtain that $M$ reads at most
$|Q|$ symbols from the factor $a^m$ until it jumps again. This time it may
jump the next $b$, where it also reads at most
$|Q|$ symbols from the factor $b^m$, or it jumps directly 
to the endmarker where a returning step is performed.
So, we have
\begin{align*}
(q_0 a^n\border^n a^n \border a^mb^m \rightend)
&\vdash^+ 
(a^{n-x_1} \border^{n-y_1} a^{n-z_1} q_4 a^mb^m \rightend)\\
&\vdash^+ 
(p_0 a^{n-x_1} \border^{n-y_1} a^{n-z_1} a^{m-\ell_1}b^{m-\ell_2} \rightend),
\end{align*}
where $0\leq \ell_1,\ell_2\leq |Q|$. 
Then, the computation 
$$
(q_0 a^n\border^n a^{z_1} \border a^{m+n-z_1}b^m \rightend)
\vdash^+
(p_0 a^{n-x_1} \border^{n-y_1} a^{n+m-z_1-\ell_1} b^{m-\ell_2} \rightend)
$$
continues accepting as well, a contradiction. This concludes Case 1.1.1
and we know that $M$ jumps over the remaining second factor $a^{n-z_1}$ either
to the first symbol $b$ or to the endmarker.

\medskip
\noindent
\textbf{Case~1.1.2.} We assume furthermore that $M$ now jumps where the 
symbol $\border$ is translucent.
Since in this case, by assumption that $M$ jumps, the symbol $a$ is
translucent as well, one possibility is that $M$ jumps to the first $b$. As
before, it may read at most $|Q|$ symbols from the factor $b^m$ until
it jumps to the endmarker and performs a returning step.
So, we obtain in the current case
\begin{align*}
(q_0 a^n\border^n a^n \border a^mb^m \rightend)
&\vdash^+ 
(a^{n-x_1} \border^{n-y_1} q_5 a^{n-z_1} \border a^m b^{m} \rightend)\\
&\vdash^+
(p_0 a^{n-x_1} \border^{n-y_1} a^{n-z_1} \border a^m b^{m-\ell_2} \rightend),
\end{align*}
where $\{a,\border\}\subseteq \tau(q_5)$, $0\leq\ell_2\leq |Q|$.
Then, the computation 
\begin{align*}
(q_0 a^n\border^{n-1} a^{z_1} \border a^{n-z_1}\border a^{m} b^m \rightend)
\vdash^+
(a^{n-x_1} \border^{n-1-y_1} q_5 \border a^{n-z_1} \border a^m b^{m} \rightend)\\
\vdash^+
(p_0 a^{n-x_1} \border^{n-1-y_1} \border a^{n-z_1} \border a^m b^{m-\ell_2} \rightend),
\end{align*}
continues accepting as well, a contradiction. This concludes Case 1.1.2
and \mbox{Case~1.1.} Now, we know that if $M$ jumps from the prefix $a^n$ to the
first $\border$ (Case~1), then its next jump is with translucent $a$.

\medskip
\noindent
\textbf{Case~1.2.} We assume that $M$ jumps from the infix $\border^n$
with translucent $a$. Dependent on whether $b\in\tau(r_1)$,
now $M$ jumps to the first $b$, where it may read at most $|Q|$ symbols from
the factor $b^m$ and then returns or it returns immediately:
\begin{align*}
(q_0 a^n\border^n a^n \border a^mb^m \rightend)
&\vdash^+ 
(a^{n-x_1}  r_1 \border^{n-y_1} a^{n} \border a^mb^m \rightend)\\
&\vdash^+ 
(r_0 a^{n-x_1} \border^{n-y_1} a^{n} \border a^m b^{m-\ell_2} \rightend),
\end{align*}
where $\{a,\border\}\subseteq \tau(r_1)$.
Then, the computation 
\begin{align*}
(q_0 a^{n-1}\border^{y_1} a \border^{n-y_1} a^{n} \border a^{m}b^m \rightend)
&\vdash^+
(a^{n-1-x_1}  r_1 a  \border^{n-y_1} a^{n} \border a^mb^m \rightend)\\
&\vdash^+ 
(r_0 a^{n-1-x_1}a \border^{n-y_1} a^{n} \border a^mb^{m-\ell_2} \rightend)
\end{align*}
continues accepting as well, a contradiction. This concludes Case 1.2
and Case~1. Now, we know that $M$ jumps from the prefix $a^n$
with translucent~$\border$. 

\medskip
\noindent
\textbf{Case~2.} We assume that $M$ jumps where the symbol $\border$ is
translucent, that is, $\{a,\border\}\subseteq \tau(q_1)$:
$$
(q_0 a^n\border^n a^n \border a^mb^m \rightend)
\vdash^* 
(q_1 a^{n-x_1} \border^n a^n \border a^mb^m \rightend)
\vdash^+
(s_0a^{n-x_1} \border^{n} a^n \border a^mb^{m-\ell_2} \rightend),
$$
where $\ell_2=0$ if $b\in \tau(q_1)$. This is one sweep over the input. We
consider the next sweep. Either it leads to one of the cases already treated
or it yields the configuration 
$$
(s_1a^{n-x_1-x_2} \border^{n} a^n \border a^mb^{m-\ell_2-\ell_3} \rightend)
$$
similar as after the first sweep. If $x_1+x_2\geq |Q|$, we obtain a contradiction
with the arguments before Case~1, namely, that $M$ has to jump from the prefix
before having read $|Q|$ consecutive symbols. Similar arguments can be used to
derive a contradiction if $\ell_2+\ell_3\geq |Q|$.
If $x_1+x_2 < |Q|$ and $\ell_2+\ell_3< |Q|$, we consider
the next sweep. Again, either it leads to one of the cases already treated
or it yields the configuration 
$$
(s_2a^{n-x_1-x_2-x_3} \border^{n} a^n \border a^mb^{m-\ell_2-\ell_3-\ell_4} \rightend).
$$
Concluding inductively, since the computation is accepting, after some sweeps 
there must occur a configuration of the cases already treated. Note that $M$
must halt for acceptance. 
This implies a contradiction also for Case~2 and concludes the proof.
\strut
\end{proof}

\begin{theorem}\label{theo:nr-not-closed-concat}
The families $\lfam(\nrdfawtl)$ and
$\lfam(\nrnfawtl)$ are not closed under concatenation.
\end{theorem}

\begin{proof}
By Lemma~\ref{lem:lc-not-nrnfawtl}, it is sufficient 
to show that $L_{c,1}$ and $L_{c,2}$ are accepted by $\nrdfawtl$.

The language $L = \{\, v\in \{a,b\}^+ \mid |v|_a=|v|_b\,\}$ is accepted by some
$\dfawtl$~\cite{Nagy:2013:dpcdsosdra}. The construction can straightforwardly
be extended to some $\dfawtl$ that accepts~$L_{c,2}$.

The language $\{\, a^nb^nc^n \mid n \geq 0\,\}$ is accepted by some
$\nrdfawtl$~\cite{mraz:2023:nrfawtl}. This construction can straightforwardly
be extended to some $\nrdfawtl$ that accepts~$L_{c,1}$.
\end{proof}

The proofs of Lemma~\ref{lem:lc-not-nrnfawtl} and Theorem~\ref{theo:nr-not-closed-concat} also
enable us to obtain the non-closure under Kleene star.

\begin{theorem}\label{theo:nr-not-closed-star}
The families $\lfam(\nrdfawtl)$ and
$\lfam(\nrnfawtl)$ are not closed under Kleene star.
\end{theorem}

\begin{proof}
Since language $L_{c,2}$ always starts with a $\border$-symbol and $L_{c,1}$ not, it is straightforward to construct an $\nrdfawtl$
accepting $L'_c=L_{c,1} \cup L_{c,2}$. 
Next, we argue that ${L'_c}^*$ is not accepted by any $\nrnfawtl$. To this end, our approach
uses almost literally the proof of Lemma~\ref{lem:lc-not-nrnfawtl} where 
an accepting computation on input $w=a^n\border^n a^n \border a^mb^m$ for $m,n$ large enough is considered.
Since $w \in {L'_c}^*$, we can use the same argumentation, but we have to take care of the
following: whenever a contradiction due to a word not in $L_c$ is obtained there, this word must not belong to ${L'_c}^*$ here.
However, inspecting the words occurring before Case~1, in Case~1, Case~1.1, Case~1.1.1, Case~1.1.2, and Case~1.2 gives that each word is not in~${L'_c}^*$
Hence, we can conclude as in the proof of Lemma~\ref{lem:lc-not-nrnfawtl} that~${L'_c}^*$ is not accepted by 
any $\nrnfawtl$.

Then, the assumption that $\lfam(\nrdfawtl)$ or $\lfam(\nrnfawtl)$ is closed under Kleene star, leads in both
cases to a contradiction and we obtain the non-closure under Kleene star.
\end{proof}

Next, we turn to the operation of reversal and use again the witness 
language~$L_c$ to show the non-closure of the family $\lfam(\nrnfawtl)$.

\begin{lemma}\label{lem:lcr-accepted-nrndawtl}
The language $L^R_c$ is accepted by some $\nrdfawtl$.
\end{lemma}

\begin{proof}
An $\nrdfawtl$ accepting $L^R_c = L^R_{c,2} \cdot L^R_{c,1}$
works in two phases. In a first phase, it verifies that the input prefix
up to the leftmost $\border$ meets the condition of~$L^R_{c,2}$. Then, on
reading the leftmost $\border$ the second phase is started. In this phase
some $\nrdfawtl$ for $L^R_{c,1}=L_{c,1}$ is simulated. 
For the sake of completeness, we construct 
an $\nrdfawtl$ $M=\langle Q,\{a,b,\border\}, q_0, \rightend, \tau,\delta\rangle$
accepting $L^R_c$ as follows.
\begin{itemize}
\item $Q=\{q_0,q_a,q_b, p_0,p_a,p_{\border},r_1,r_2\}$,\smallskip
\item 
$\tau(q_0)=\tau(p_0)=\emptyset$, $\tau(q_a)=\tau(p_a)=\{a\}$, $\tau(q_b)=\{b\}$,\\[1mm]
  $\tau(p_\border)=\{\border\}$, $\tau(r_1)=\{a,b,\border\}$, $\tau(r_2)=\{a\}$.
\end{itemize}

\begin{center}
\renewcommand{\arraystretch}{1.1}\tabcolsep2pt
\begin{tabular}[t]{rccl}
(1) &  $\delta(q_0,a)$ &=& $q_a$,\\
(2) &  $\delta(q_0,b)$ &=& $q_b$,\\
(3) &  $\delta(q_0,\border)$ &=& $p_0$,\\
(4) &  $\delta(q_a,b)$ &=& $r_1$,\\
(5) &  $\delta(q_b,a)$ &=& $r_1$,\\
(6) &  $\delta(r_1,\rightend)$ &=& $q_0$,\\
\end{tabular}
\qquad\qquad
\begin{tabular}[t]{rccl}
(7) &  $\delta(p_0,a)$ &=& $p_a$,\\
(8) &  $\delta(p_a,\border)$ &=& $p_\border$,\\
(9) &  $\delta(p_\border,a)$ &=& $r_2$,\\
(10) &  $\delta(r_2,\rightend)$ &=& $p_0$,\\
(11) &  $\delta(p_0,\rightend)$ &=& $\accept$.\\
\end{tabular}
\end{center}

In the first phase, $M$ reads the first input symbol (Transitions~1 and~2),
memorizes it in its state, and jumps to the next matching symbol
(Transitions~4 and~5). The jump is with $\border$ visible. So the matching
symbols have to appear before the leftmost $\border$. Then state $r_1$ is
entered with all symbols translucent, so that a return step is caused
(Transition~6). When all symbol up to the first $\border$ are matched, the
$\border$ is read and the prefix is verified to belong to~$L^R_{c,2}$. In case
of any mismatch, the computation gets stuck and rejects. Now the second phase
starts. This phase consists of several sweeps. At the beginning of a sweep,
an $a$ is read and memorized in the state (Transition~7). Then only $a$'s become
translucent and the leftmost $\border$ is read and memorized in the state
(Transition~8). Now, only~$\border$'s become
translucent and the next $a$ is read (Transition~9). 
So far, in the current sweep one symbol
is read from each of the three blocks. Being in state~$r_2$ only symbol $a$ is
translucent. If the next visible symbol is the endmarker, the input to the
second phase is verified to be of the form $a^+\border^+ a^+$. Next, a return
step is performed (Transition~10) and the next sweep starts. At the beginning
of each sweep, $M$ is in state $p_0$ for which all symbols are visible. So,
if and only if the symbol seen at the beginning of a sweep is the endmarker,
the input to the second phase is verified to belong to~$L^R_{c,1}$. Only in this
case, $M$ can accept (Transition~11).
\end{proof}

Lemma~\ref{lem:lc-not-nrnfawtl} and Lemma~\ref{lem:lcr-accepted-nrndawtl}
imply the next corollary.

\begin{corollary}\label{cor:nrnfa-not-closed-rev}
The families $\lfam(\nrdfawtl)$ and $\lfam(\nrnfawtl)$
are not closed under reversal.
\end{corollary}

\subsection{Non-Closure under Inverse Homomorphism}\label{subsec:non-closed-invhom}

In this section, we obtain the result that all four language families studied in this paper are not
closed under inverse non-erasing homomorphism. We look first on the returning case and note that
it is shown in~\cite{nagy:2012:oncdsosdrawwso} in terms of CD-systems of stateless deterministic restarting automata 
that the language family $\lfam(\nfawtl)$ is closed under inverse projections. However, this result cannot be 
generalized, since we will show the non-closure under inverse non-erasing homomorphism in the following. We start
with presenting a language $L_{cd}$ that is accepted by a $\dfawtl$ and will be used later to obtain the non-closure result.
Let the language $L_{cd}\subseteq\{a, b, c, d, \border\}^*$ be 
defined as
$$
L_{cd} = \{\, u\border v \mid u \in \{a,c\}^*,\, v \in \{b,c,d\}^*,\, |uv|_c=|v|_d \,\}.
$$ 

\begin{lemma}\label{lem:invhom1:pos}
Language $L_{cd}$ belongs to $\lfam(\dfawtl)$.
\end{lemma}

\begin{proof}
We are now constructing a $\dfawtl$ $M$ for $L_{cd}$. 
Intuitively, the behavior of $M$ on an input string 
$u\border v\in L_{cd}$ develops along the following phases:
\begin{enumerate}
\item Only $c$'s are translucent and only $a$ symbols are processed as long as the $\border$-symbol is read.
From the next phase on $a$ symbols are visible and $M$ rejects the input whenever an~$a$ is read. This ensures
that the input is of the form $\{a,c\}^*\border \{b,c,d\}^*$.

\item A symbol $c$ is read while $b$'s and $d$'s are translucent and matched against a $d$ symbol
while $b$'s and~$c$'s are translucent. This phase is repeated until no more $c$ is found.

\item The input is accepted if the remaining input consists of $b$ symbols only.
\end{enumerate}

To formally implement the above three phases we define a $\dfawtl$
$M= \langle Q, \{a,b,c,d,\border\}, q_0, \rightend, \tau, \delta \rangle.$
 We~let $Q=\{q_0,q_1,q_2,q_3\}$, the translucency function $\tau(q_0)=\{c\}$, $\tau(q_1)=\{b,d\}$,\ $\tau(q_2)=\{b,c\}$,\ $\tau(q_3)=\{b\}$,
and the transition function $\delta$~as being:
\begin{center}
\renewcommand{\arraystretch}{1.1}
\begin{tabular}[t]{rl}
(1) & $\delta(q_0,a)=q_0$\\
(2) & $\delta(q_0,\border)=q_1$\\
(3) & $\delta(q_1,c)=q_2$\\
\end{tabular}
\qquad\qquad
\begin{tabular}[t]{rl}
(4) & $\delta(q_2,d)=q_1$\\
(5) & $\delta(q_1,\rightend)=q_3$\\
(6) & $\delta(q_3,\rightend)=\accept$\\
\end{tabular}
\end{center}
Here, transitions (1) and (2) describe the first phase, transitions (3)--(5) describe the second phase,
and the last phase is realized by transition~(6). Thus, we can conclude that the $\dfawtl$ $M$ accepts $L_{cd}$.
\end{proof}

\begin{theorem}\label{thm:invhom1:neg}
The families $\lfam(\dfawtl)$ and $\lfam(\nfawtl)$ are not closed under inverse non-erasing homomorphism.
\end{theorem}

\begin{proof}
We consider language $L_{cd}$ that belongs to $\lfam(\dfawtl)$ due to Lemma~\ref{lem:invhom1:pos}
and define a homomorphism $h\colon \{a,b,\border\}^* \to \{a,b,c,d,\border\}^*$ such that $h(a)=ac$, $h(b)=bd$, and $h(\border)=\border$.
Let us first assume that $\lfam(\dfawtl)$ is closed under inverse homomorphism.
Then, the inverse homomorphic image $h^{-1}(L_{cd})=\{\, a^n\border b^m \mid n = m \,\}$ belongs to $\lfam(\dfawtl)$ as well.
It is shown in~\cite{nagy:2012:oncdsosdrawwso} that every language in $\lfam(\nfawtl)$ and, hence, also every language in $\lfam(\dfawtl)$
contains a letter-equivalent regular sublanguage.
However, language $\{\, a^n\border b^m \mid n = m \,\}$ does not contain any letter-equivalent regular sublanguage
which can be shown by a simple application of the pumping lemma for regular languages. This is a contradiction and
gives the non-closure of $\lfam(\dfawtl)$ under inverse non-erasing homomorphism. The proof for the family $\lfam(\nfawtl)$
can be given analogously.
\end{proof}

Next, we investigate the non-returning variants and show the non-closure under inverse non-erasing homomorphism as well.
We start with presenting a language $L_{ac}$ that is accepted by an $\nrdfawtl$ and will be subsequently used to obtain the non-closure result.
Let the language $L_{ac}\subseteq\{a, c, \border\}^*$ be 
defined as
$$
L_{ac} = \{\, u\border v \mid u,v \in \{a,c\}^*,\, |u|_a=|v|_c \,\}.
$$ 

\begin{lemma}\label{lem:invhom2:pos}
Language $L_{ac}$ belongs to $\lfam(\nrdfawtl)$.
\end{lemma}

\begin{proof}
We construct the $\nrdfawtl$
$M= \langle Q, \{a,c,\border\}, q_0, \rightend, \tau, \delta \rangle$ as follows.
Let $Q=\{q_0,q_1,q_2,q_3,q_4\}$, the translucency function $\tau(q_0)=\{a\}$, $\tau(q_1)=\{c\}$,\ $\tau(q_2)=\emptyset$, $\tau(q_3)=\{a\}$, $\tau(q_4)=\{c\}$,
and the transition function $\delta$~as being:
\begin{center}
\renewcommand{\arraystretch}{1.1}
\begin{tabular}[t]{rl}
(1) & $\delta(q_0,c)=q_0$\\
(2) & $\delta(q_0,\border)=q_1$\\
(3) & $\delta(q_1,a)=q_1$\\
(4) & $\delta(q_1,\rightend)=q_2$\\
\end{tabular}
\qquad\qquad
\begin{tabular}[t]{rl}
(5) & $\delta(q_2,a)=q_3$\\
(6) & $\delta(q_3,c)=q_4$\\
(7) & $\delta(q_4,\rightend)=q_2$\\
(8) & $\delta(q_2,\rightend)=\accept$\\
\end{tabular}
\end{center}

The automaton basically works in two phases.
In the first phase, it makes one sweep from left to right reading all $c$'s to the left of the $\border$-symbol and all $a$'s to the right of
the $\border$-symbol by using transitions~(1) to~(4). After this sweep the remaining input is of the form $a^*c^*$.
The second phase consists of iterated sweeps, whereby in every sweep one symbol~$a$ is matched against one symbol~$c$ by using transitions~(5) to~(7).
If all $a$'s are matched against all $c$'s and no input symbols are left, the automaton accepts by using transition~(8).
Thus, we can conclude that the $\nrdfawtl$ $M$ accepts $L_{ac}$.
\end{proof}
 
\begin{theorem}\label{lem:invhom2:neg}
The families $\lfam(\nrdfawtl)$ and $\lfam(\nrnfawtl)$ are not closed under inverse non-erasing homomorphism.
\end{theorem}

\begin{proof}
We consider language $L_{ac}$ that belongs to $\lfam(\nrdfawtl)$ due to Lemma~\ref{lem:invhom2:pos}
and define a non-erasing homomorphism $h\colon \{b,\border\}^* \to \{a,c,\border\}^*$ such that $h(b)=ac$ and $h(\border)=\border$.
To show the theorem it remains for us to show that the inverse homomorphic image
$h^{-1}(L_{ac})=\{\, b^n\border b^n \mid n \ge 1 \,\}$ does not belong to $\lfam(\nrnfawtl)$.

\begin{sloppypar}
By way of contradiction, let us assume that $\{\, b^n\border b^n \mid n \ge 1 \,\}$ is accepted by some
$\nrnfawtl$ \mbox{$M=\langle Q,\Sigma, q_0, \rightend, \tau,\delta\rangle$.}
We consider an accepting computation of~$M$ on input $b^n \border b^n$ for $n$ being large enough.
\end{sloppypar}

Suppose that $M$ reads more than $|Q|$ many symbols $b$, say $n_1$, from the
prefix of the input. Then, during the first $n_1$ steps one state, say $p_1$,
appears at least twice:
$$
(q_0 b^n\border b^n \rightend)
\vdash^* 
(p_1 b^{n-i_1} \border b^n \rightend)
\vdash^+
(p_1 b^{n-i_1-i_2} \border b^n \rightend),
$$
where $i_1,i_2\leq |Q|$.
Then, the computation 
$$
(q_0 b^{n-i_2}\border b^n \rightend)
\vdash^* 
(p_1 b^{n-i_2-i_1} \border b^n \rightend)
$$
continues accepting as well, a contradiction. So, $M$ jumps after at most
$|Q|$ steps on the prefix. If it jumps to the endmarker, the effect is a state
change only. Hence, the only remaining possibility is that $M$ jumps where the symbol~$\border$ 
is not translucent. We obtain

$$
(q_0 b^n\border b^n \rightend)
\vdash^* 
(q_1 b^{n-x_1} \border b^n \rightend)
\vdash
(b^{n-x_1} q_2 b^n \rightend).
$$
Now, $M$ may jump to the endmarker or may read $b$'s. In the latter case,
we can repeat the argument from above. If $M$ reads more than $|Q|$
consecutive symbols~$b$ some state, say $p_2$,  appears twice:
\begin{align*}
(q_0 b^n\border b^n \rightend)
&\vdash^+
(b^{n-x_1} q_2 b^n \rightend)\\
&\vdash^*
(b^{n-x_1} p_2 b^{n-j_1} \rightend)\\
&\vdash^+
(b^{n-x_1} p_2 b^{n-j_1-j_2} \rightend),
\end{align*}
where $1\leq j_1,j_2\leq |Q|$.
Then, the computation 
$$
(q_0 b^n\border b^{n-j_2} \rightend)
\vdash^*
(b^{n-x_1} p_2 b^{n-j_1-j_2} \rightend)
$$
continues accepting as well, a contradiction. So, $M$ jumps after at most
$|Q|$ steps on the second factor~$b^n$ to the endmarker.
Thus, we end up in a configuration $(q_3b^{n-x_1}b^{n-x_2} \rightend)$ from which an accepting configuration can be
reached.
But then the computation
$$
(q_0 b^{n+1}\border b^{n-1} \rightend)
\vdash^* 
(q_3 b^{n+1-x_1} b^{n-1-x_2} \rightend)=(q_3 b^{n-x_1} b^{n-x_2} \rightend)
$$
continues accepting as well which is a contradiction.

Thus, $\{\, b^n\border b^n \mid n \ge 1 \,\}$ does not belong to $\lfam(\nrnfawtl)$ which concludes the proof.
\end{proof}

For the sake of completeness and reader's ease of mind, we summarize in the following table
the scenario of closure properties for the language families
$\lfam(\dfawtl)$, $\lfam(\nrdfawtl)$, $\lfam(\nfawtl)$, and $\lfam(\nrnfawtl)$. The question
of whether the family $\lfam(\nfawtl)$ is closed under reversal remains open.

\begin{table}[!ht]
\begin{center}
\renewcommand{\arraystretch}{1.2}\setlength{\tabcolsep}{4.5pt}
\begin{tabular}{|c|c!{\color{black}\vrule width .5pt}c!{\color{black}\vrule width .5pt}c!{\color{black}\vrule width .5pt}c|c|c|c|c|c|}
\hline
\rowcolor[HTML]{EFEFEF} 
  Family & {$\overline{\phantom{aa}}$} & {$\cup$} & ${\cap}$ &  
           $\cap_\subtext{reg}$ & $\cdot$ & $*$ & $h_{\lambda}$ &$h_{\lambda}^{-1}$ & $R$\\
\hline\hline
\cellcolor[HTML]{EFEFEF}$\lfam(\dfawtl)$   
               &  \tcyes~\cite{mraz:2024:rfawtl:proc} 
               & \tcno~\cite{mraz:2024:rfawtl:proc}  
               & \tcno~\cite{mraz:2024:rfawtl:proc}  
               & \tcno~\cite{mraz:2024:rfawtl:proc} 
               & \tcno~\cite{mraz:2024:rfawtl:proc}
               & \tcno~\cite{mraz:2024:rfawtl:proc}
               & \tcno~\cite{mraz:2024:rfawtl:proc} 
               & \cno
               & \tcno~\cite{mraz:2024:rfawtl:proc}\\
\cellcolor[HTML]{EFEFEF}$\lfam(\nrdfawtl)$   
              & \tcyes~\cite{mraz:2023:nrfawtl}  
              & \tcno~\cite{mraz:2023:nrfawtl}  
              & \tcno~\cite{mraz:2023:nrfawtl}  
              & \tcno~\cite{kutrib:2025:twfawtil:proc}
              & \cno\ (\copen~\cite{mraz:2023:nrfawtl}) 
              & \cno\ (\copen~\cite{mraz:2023:nrfawtl}) 
              & \tcno~\cite{mraz:2023:nrfawtl}   
              & \cno\ (\copen~\cite{mraz:2023:nrfawtl}) 
              & \cno\\
\cellcolor[HTML]{EFEFEF}$\lfam(\nfawtl)$   
              & \tcno~\cite{nagy:2012:oncdsosdrawwso} 
              & \tcyes~\cite{nagy:2012:oncdsosdrawwso} 
              & \tcno~\cite{nagy:2012:oncdsosdrawwso}  
              & \tcno~\cite{nagy:2012:oncdsosdrawwso} 
              & \tcyes~\cite{nagy:2012:oncdsosdrawwso}
              & \tcyes~\cite{nagy:2012:oncdsosdrawwso} 
              & \tcno~\cite{nagy:2012:oncdsosdrawwso} 
              & \cno\ (\copen~\cite{nagy:2012:oncdsosdrawwso})
              & \copen~\cite{nagy:2012:oncdsosdrawwso}\\
\cellcolor[HTML]{EFEFEF}$\lfam(\nrnfawtl)$   
             & \cno
             & \tcyes~\cite{mraz:2023:nrfawtl} 
             & \tcno~\cite{kutrib:2025:twfawtil:proc}
             & \tcno~\cite{kutrib:2025:twfawtil:proc}
             & \cno\ (\copen~\cite{mraz:2023:nrfawtl}) 
             & \cno\ (\copen~\cite{mraz:2023:nrfawtl})
             & \tcno~\cite{mraz:2023:nrfawtl} 
             & \cno\ (\copen~\cite{mraz:2023:nrfawtl}) 
             & \cno\\ 
\hline
\end{tabular}
\end{center}
\caption{A summary of closure properties for the language families
  $\lfam(\dfawtl)$, $\lfam(\nrdfawtl)$, $\lfam(\nfawtl)$, and
  $\lfam(\nrnfawtl)$. Non-erasing homomorphisms are denoted $h_{\lambda}$ and
  $R$ means reversal.
Bold properties are shown in the present paper.
Shaded properties are proved \mbox{in \cite{kutrib:2025:twfawtil:proc}, \cite{mraz:2023:nrfawtl}, \cite{mraz:2024:rfawtl:proc}, 
and ~\cite{nagy:2012:oncdsosdrawwso}.}}
\label{tab:closure}
\end{table}

\section{The emptiness Problem}\label{sec:open-decide}

To decide the emptiness problem for non-returning $\dfa$s and $\nfa$s with translucent letters,
we make use of iterated finite-state transducers. The idea is to
simulate computations of translucent automata by such transducers.
In particular, we construct an iterated transducer that mimics the
behavior of a translucent $\dfa$ or $\nfa$ performing a bounded
number of sweeps. This approach allows us to exploit known results
on iterated transducers. In particular, it has been shown that
iterated length-preserving transducers with a bounded number of
sweeps accept only regular languages~\cite{kutrib:2022:cadponiufst}. By reducing
the behavior of translucent automata to such devices, we obtain a
regular representation of their computations, which will later
enable us to decide the emptiness problem.

  An \emph{iterated uniform finite-state transducer} is a finite-state
  transducer that processes its input in multiple sweeps. In the first pass, the machine reads the input word
  followed by an endmarker and produces an output word. In each
  subsequent pass, it reads the output generated in the previous pass
  and produces a new output word. Thus, the output of one pass serves
  as the input to the next pass. In~\cite{kutrib:2022:cadponiufst}
  length-preserving iterated uniform finite-state transducers are considered that are also known as Mealy
  machines~\cite{Mealy:1955:amfssc}. It has been shown that, if the number of
  sweeps of these devices is limited by a constant, then the language accepted is
regular~\cite{kutrib:2022:cadponiufst}.

  Formally, we define a \emph{nondeterministic iterated uniform
  finite-state transducer} ($\ntrans$) as a system
  $T=\langle Q,\Sigma,\Delta,q_0,\rightend,\delta,F\rangle$,
  where:
  \begin{itemize}
  \item $Q$ is the finite set of \emph{internal states},
  \item $\Sigma$ is the set of \emph{input symbols},
  \item $\Delta$ is the set of \emph{output symbols},
  \item $q_0\in Q$ is the initial state,
  \item $\rightend\in\Delta\setminus \Sigma$ is the \emph{endmarker},
  \item $F\subseteq Q$ is the set of \emph{accepting states},
  \item $\delta\colon Q\times(\Sigma\cup\Delta)\to 2^{Q\times \Delta}$
  is the partial \emph{transition function}.
  \end{itemize}

  The $\ntrans$ $T$ \emph{halts} whenever the transition function is
  undefined or when it enters an accepting state at the end of a sweep.
  Since the transducer is applied in multiple passes, that is, in every
  pass except the first it processes the output of the previous pass,
  the transition function depends on symbols from $\Sigma\cup\Delta$.
  For $w\in(\Sigma\cup\Delta)^*$ we denote by $T(w)$ the set of outputs
  that $T$ may produce in a complete sweep on input $w$. During a
  computation on input $w\in\Sigma^*$, the $\ntrans$ $T$ produces a
  sequence of words
  $w_1,\ldots,w_i,w_{i+1},\ldots\in(\Sigma\cup\Delta)^*$ such that
  $w_1\in T(w\rightend)$ and $w_{i+1}\in T(w_i)$ for all $i\geq 1$.

  An $\ntrans$ is said to be \emph{deterministic} if
  $|\delta(p,x)|\leq 1$ for all $p\in Q$ and $x\in\Sigma\cup\Delta$.
  In this case we write $\delta(p,x)=(q,y)$ instead of
  $\delta(p,x)=\{(q,y)\}$ and treat the transition function as a
  mapping $\delta\colon Q\times(\Sigma\cup\Delta)\to Q\times \Delta$.

  A computation is halting if there exists $r\geq 1$ such that
  $T$ halts on $w_r$, thus performing $r$ sweeps. An input word
  $w\in\Sigma^*$ is \emph{accepted} by $T$ if at least one computation
  on $w$ halts at the end of a sweep in an accepting state. That is,
  the initial input consists of a word over $\Sigma$ followed by the
  endmarker, and after $r-1$ sweeps an output is produced that drives
  $T$ through a final complete sweep where it halts in an accepting
  state. The output of the last sweep is not used. The language
  accepted by $T$ is the set
  $L(T)=\{\, w \in \Sigma^* \mid w \text{ is accepted by }T\,\}$.

  In the following, we construct, for every $\nrdfawtl$
  (resp.\ $\nrnfawtl$)
  $M=\langle Q,\Sigma,q_0,\rightend,\tau,\delta\rangle$
  that performs at most $k$ sweeps, a corresponding $\ntrans$
  $M'=\langle Q',\Sigma,\Delta,q_0',\rightend,\delta',F'\rangle$
  with $L(M)=L(M')$.

  \begin{proposition}\label{prop:it-trans}
  Let $k\geq 1$ be a positive integer. Given an $\nrnfawtl$ $M=\langle
  Q,\Sigma, q_0, \rightend, \tau,\delta\rangle$ that makes at most $k$ sweeps,
  then there is an $\ntrans$ 
  $T=\langle Q',\Sigma,\Delta,q_0,\rightend,\delta',F\rangle$ $M'$ that makes
  at most $k$ sweeps such $L(M)=L(T)$.
  \end{proposition}

  \begin{proof}
  By definition, $M$ accepts only at the end of a sweep. 
  To simulate the at most $k$ sweeps of $M$, the machine $T$ uses
  $k$ copies of every state $q\in Q$, denoted by $q_1,q_2,\dots,q_k$,
  where $q_i$ represents state $q$ during the $i$th sweep.
 So, $Q'=\{\ q_i \mid q\in Q \text{ and } 1\leq i\leq k\,\} \cup\{s_+\}$. We set
 $\Delta=\Sigma\cup\{\rightend,\blank\}$ and $F=\{s_+\}$.
 For all $1\leq i\leq k$, $q\in Q$, the transition function $\delta'$ is defined by:
  \begin{align*}
  \delta'(q_i,x)&\ni (q'_i,\blank)
  &&\text{if } \delta(q,x) \ni q',\ x\in \Sigma \setminus \tau(q),\\
  \delta'(q_i,y) &\ni (q_i,y)
  &&\text{if } y\in\tau(q),\\
  \delta'(q_i,\blank) &\ni (q_i,\blank)
  &&\\
  \delta'(q_i,\rightend) &\ni (q'_{i+1},\rightend)
  &&\text{if } \delta(q,\rightend)\ni q'.
  \end{align*}
  Transducer $M'$ accepts exactly when $M$ accepts:
  \begin{align*}
  \delta'(q_i,\rightend ) &\ni (s_+,\rightend)
  &&\text{if } \delta(q,\rightend) \ni \accept.\\
  \end{align*}
  Consider the computation
  \[
  (w_0qw_1xw_2\rightend)
  \vdash (w_0w_1q'w_2\rightend) \vdash^+ \accept
  \]
  of $M$ with $\delta (q,x)=q'$ and $w_1\in (\tau (q))^*$. Then 
  $T$ can perform the transduction
  \[
  (w_0'qw_1'xw_2'\rightend)
  \vdash (w_0'w_1'\blank q'w_2'\rightend) \vdash^+ (w_0'w_1'\blank
    w_2'\rightend s_+)
  \]
  where $w_0', w_1', w_2'$ are the words $w_0, w_1, w_2$ where
  each letter read by $M$ in predecessor steps is replaced by a blank symbol ($\blank$). 
  Since all translucent letters are
  read and emitted by $T$ and the blank symbols are ignored during the
  subsequent sweeps, $M$ and $T$ accept the same languages $L(T)=L(M)$.
  \end{proof}

  Since iterated finite-state transducers with a constant number of
  sweeps accept only regular languages~\cite{kutrib:2022:cadponiufst}, the
  same holds for 
$\nrdfawtl$ and $\nrnfawtl$ with a constant number of sweeps.

  \begin{theorem}\label{theo:empty-dec}
  Given an $\nrdfawtl$ ($\nrnfawtl$) $M=\langle Q,\Sigma, q_0, \rightend, \tau,\delta\rangle$, it is decidable whether the language accepted, $L(M)$, is empty.
  \end{theorem}

  \begin{proof}

  Assume the computation of $M$ on a word
  \[
  w = x_0 a_0 u_0 x_1 a_1 u_1 \cdots x_l a_l u_l
  \]
  is accepting, where $x_i, a_i, u_i \in \Sigma^*, 0\le i\le l$.
  Assume furthermore that there are two sweeps in the computation, started by
  $M$ in the same state, say $s \in Q$, such that at least one symbol is read
  in between:
  \[
  \begin{aligned}
  (q_0 x_0a_0u_0x_1a_1u_1\cdots x_la_lu_l \rightend)
  &\vdash^* (s a_0u_0a_1u_1\cdots a_lu_l \rightend) \\
  &\vdash^+ (s u_0u_1\cdots u_l \rightend)
  \vdash^+ \accept 
  \end{aligned}
  \]
where $a_0a_1\cdots a_l\neq \lambda$ and the factors $x_i$ are read in the sweeps until configuration 
$(s a_0u_0a_1u_1\cdots a_lu_l \rightend)$ is reached. In this phase the $a_i$
and $u_i$ are translucent in any sweep. Similarly, the factors $a_i$ are read
in the next sweeps until configuration 
$(s u_0u_1\cdots u_l \rightend)$ is reached, where the $u_i$
are translucent in any sweep.
  Then the computation starting in configuration
  \[
  (q_0 x_0 u_0 x_1 u_1 \cdots x_l u_l\rightend)
  \]
  is accepting as well. Indeed, before reaching configuration $(s
  a_0u_0a_1u_1\cdots a_lu_l \rightend)$, automaton $M$ is not aware of the
  presence of the symbols $a_0,a_1,\dots,a_l$. Hence, there must exist a computation
  \[
  \begin{aligned}
  (q_0 x_0u_0x_1u_1\cdots x_lu_l \rightend)
  &\vdash^* (s u_0u_1\cdots u_l \rightend)
  \vdash^+ \accept .
  \end{aligned}
  \]

We can safely assume that in any accepting computation there is no useful loop
on sweeps that do not consume any letter. Thus, if there is a loop in the
sequence of sweeps, we can always find a shorter word that is accepted as
well. This implies that if $L(M)$ is non-empty then there is a word in $L(M)$ that
is accepted with at most $|Q|$ sweeps. A simple modification that adds a
finite counter to the states yields an $\nrdfawtl$ $M'$ that accepts exactly
the words of $L(M)$ accepted by $M$ in at most $|Q|$ sweeps. In particular,~$L(M)$
is empty if and only if $L(M')$ is empty.

By Proposition~\ref{prop:it-trans}, any $\nrnfawtl$ restricted to at most 
$|Q|$ sweeps can effectively be converted into an equivalent $\ntrans$ 
that makes at most $|Q|$ sweeps either. We do so with $M'$ and obtain 
an $\ntrans$ $T$. Then, the language accepted by $T$ is
regular and can effectively be constructed from $T$~\cite{kutrib:2022:cadponiufst}.
As emptiness is decidable for regular
  languages, it follows that the emptiness problem is also decidable
  for $\nrdfawtl$ and $\nrnfawtl$.
  \end{proof}


\end{document}